\documentclass{llncs}
\usepackage{amssymb}
\usepackage{graphicx}

\usepackage{url}
\urldef{\mails}\path|meizeh@cs.technion.ac.il|

\usepackage{ntheorem}
\theoremseparator{.}
\newtheorem{obs}{Observation}

\newtheorem{cor}{Corollary}

\newtheorem{reducerule}{Reduction Rule}
\newtheorem{branchrule}[reducerule]{Branching Rule}

\newcommand{\myparagraph}[1]{\par\smallskip\par\noindent{\bf{}#1:~}}
\newcommand{\comment}[1]{}

\usepackage{algorithm}
\usepackage{algorithmic}
\newcommand{\alg}[1]{\mbox{\sf #1}}  

\usepackage{amsopn}

\begin{document}

\mainmatter

\title{Maximization Problems Parameterized\\Using Their Minimization Versions:\\The Case of Vertex Cover}

\author{Meirav Zehavi}

\institute{Department of Computer Science, Technion - Israel Institute of Technology,
 Haifa 32000, Israel\\
\mails}

\maketitle

\begin{abstract}
The parameterized complexity of problems is often studied with respect to the size of their optimal solutions. However, for a maximization problem, the size of the optimal solution can be very large, rendering algorithms parameterized by it inefficient. Therefore, we suggest to study the parameterized complexity of maximization problems with respect to the size of the optimal solutions to their {\em minimization} versions. We examine this suggestion by considering the {\sc Maximal Minimal Vertex Cover (MMVC)} problem, whose minimization version, {\sc Vertex Cover}, is one of the most studied problems in the field of Parameterized Complexity. Our main contribution is a parameterized approximation algorithm for {\sc MMVC}, including its weighted variant. We also give conditional lower bounds for the running times of algorithms for {\sc MMVC} and its weighted variant.
\end{abstract}

\section{Introduction}

The parameterized complexity of problems is often studied with respect to the size of their optimal solutions. However, for a maximization problem, the size of the optimal solution can be very large, rendering algorithms parameterized by it inefficient. Therefore, we suggest to study the parameterized complexity of maximization problems with respect to the size of the optimal solutions to their {\em minimization} versions. Given a maximization problem, the optimal solution to its minimization version might not only be significantly smaller, but it might also be possible to efficiently compute it by using some well-known parameterized algorithm---in such cases, one can know in advance if for a given instance of the maximization problem, the size of the optimal solution to the minimization version is a good choice as a parameter. Furthermore, assuming that an optimal solution to the minimization version can be efficiently computed, one may use it to solve the maximization problem; indeed, the optimal solution to the maximization problem and the optimal solution to its minimization version may share useful, important properties.

We examine this suggestion by studying the {\sc Maximal Minimal Vertex Cover (MMVC)} problem. This is a natural choice---the minimization version of {\sc MMVC} is the classic {\sc Vertex Cover (VC)} problem, one of the most studied problems in the field of Parameterized Complexity. In {\sc Weighted MMVC (WMMVC)}, we are given a graph $G=(V,E)$ and a weight function $w: V\rightarrow\mathbb{R}^{\geq 1}$. We need to find the maximum weight of a set of vertices that forms a {\em minimal} vertex cover of $G$. The {\sc MMVC} problem is the special case of {\sc WMMVC} in which $w(v)=1$ for all~$v\in V$.

\myparagraph{Notation} Let $vc$ ($vc_w$) denote the size (weight) of a minimum vertex cover (minimum weight vertex cover) of $G$, and let $opt$ ($opt_w$) denote the size (weight) of a maximal minimal vertex cover (maximal-weight minimal vertex cover) of $G$. Clearly, $vc\leq \min\{vc_w,opt\}\leq\max\{vc_w,opt\}\leq opt_w$. Observe that the gap between $vc_w=\max\{vc,vc_w\}$ and $opt=\min\{opt,opt_w\}$ can be very large. For example, in the case of {\sc MMVC}, if $G$ is a star, then $vc_w=1$ while $opt=|V|-1$.

The standard notation $O^*$ hides factors polynomial in the input size. A problem is {\em fixed-parameter tractable (FPT)} with respect to a parameter $k$ if it can be solved in time $O^*(f(k))$, for some function $f$. The maximum degree of a vertex in $G$ is denoted by $\Delta$. Given $v\in V$, $N(v)$ denotes the set of neighbors of $v$.
Also, let $\gamma$ denote the smallest constant such that it is known how to solve {\sc VC} in time $O^*(\gamma^{vc})$, using polynomial space. Currently, $\gamma < 1.274$ \cite{vc2010}.
Given $U\subseteq V$, let $G[U]$ denote the subgraph of $G$ induced by $U$. Finally, let $M$ denote the maximum weight of a vertex in $G$.

\subsection{Related Work}

\myparagraph{MMVC} Boria et al.~\cite{MaxMinVC} show that {\sc MMVC} is solvable in time $O^*(3^m)$ and polynomial space, and that {\sc WMMVC} is solvable in time and space $O^*(2^{tw})$, where $m$ is the size of a maximum matching of $G$, and $tw$ is the treewidth of $G$. Since $\max\{m,tw\}\leq vc$ (see, e.g., \cite{paramEco}), this shows that {\sc WMMVC} is FPT with respect to $vc$. Moreover, they prove that {\sc MMVC} is solvable in time $O^*(1.5874^{opt})$ and polynomial space, where the running time can be improved if one is interested in approximation.\footnote{For example, they show that one can guarantee the approximation ratios $0.1$ and $0.4$ in times $O^*(1.162^{opt})$ and $O^*(1.552^{opt})$, respectively.} Boria et al.~\cite{MaxMinVC} also prove that for any constant $\epsilon>0$, {\sc MMVC} is inapproximable within ratios $O(|V|^{\epsilon-\frac{1}{2}})$ and $O(\Delta^{\epsilon-1})$, unless P=NP. They complement this result by proving that {\sc MMVC} is approximable within ratios $|V|^{-\frac{1}{2}}$ and $\frac{3}{2\Delta}$ in polynomial time.

Recently, Bonnet and Paschos \cite{corrSparse} and Bonnet et al.~\cite{corrMaxMinVC} obtained results relating to the inapproximability of {\sc MMVC} in subexponential time. Furthermore, Bonnet et al.~\cite{corrMaxMinVC} prove that for any $1<r\leq |V|^{\frac{1}{2}}$, {\sc MMVC} is approximable within ratio $\frac{1}{r}$ in time $O^*(2^{\frac{|V|}{r^2}})$.

\myparagraph{Vertex Cover} First, note that {\sc VC} is one of the first problems shown to be FPT. In the past two decades, it enjoyed a race towards obtaining the fastest parameterized algorithm (see \cite{vc1993,vc1995,vc1998,vc1999a,vc1999b,vc2001,wvc2003,vc2005,vc2007,vc2010}). The best parameterized algorithm, due to Chen et al.~\cite{vc2010}, has running time $O^*(1.274^{vc})$, using polynomial space. In a similar race \cite{vc2001,3vc2000,3vc2005,3vc2009,3vc2010,3vc2013}, focusing on the case where $\Delta=3$, the current winner is an algorithm by Issac et al.~\cite{3vc2013}, whose running time is $O^*(1.153^{vc})$. For {\sc Weighted VC (WVC)}, parameterized algorithms were given in \cite{wvc2003,wvc2006,wvc2009,corrweighted}. The fastest ones (in \cite{corrweighted}) use time $O^*(1.381^{vc_w})$ and polynomial space, time and space $O^*(1.347^{vc_w})$, and time $O^*(1.443^{vc})$ and polynomial space.

Kernels for {\sc VC} and {\sc WVC} were given in \cite{vc2001,wvcker2008,wvcker2013}, and results relating to the parameterized approximability of {\sc VC} were also examined in the literature (see, e.g., \cite{parApproxVC3,parApproxVC2,parApproxVC1}). Finally, in the context of Parameterized Complexity, we would also like to note that $vc$ is a parameter of interest; indeed, apart from {\sc VC}, there are other problems whose parameterized complexity was studied with respect to this parameter (see, e.g., \cite{paramVC2,paramVC1,paramVC3}).

\subsection{Our Contribution}

While it can be easily seen that {\sc WMMVC} is solvable in time $O^*(2^{vc})$ and polynomial space (see Section \ref{sec:simple}), we observe that this result might be essentially tight. More precisely, we show that even if $G$ is a bipartite graph and $w(v)\in\{1,1+\frac{1}{|V|}\}$ for all $v\in V$, an algorithm that solves {\sc WMMVC} in time $O^*((2-\epsilon)^{vc_w})$, which upper bounds $O^*((2-\epsilon)^{vc})$, contradicts the SETH (Strong Exponential Time Hypothesis). We also show that even if $G$ is a bipartite graph, an algorithm that solves {\sc MMVC} in time $O^*((2-\epsilon)^{\frac{vc}{2}})$ contradicts the SETH.

Then, we turn to present our main contribution, \alg{ALG}, which is a parameterized approximation algorithm for {\sc WMMVC} with respect to the parameter $vc$. More precisely, we prove the following theorem, where $\alpha$ is a user-controlled parameter that corresponds to a tradeoff between time and approximation ratio.

\begin{theorem}\label{theorem:approx}
For any $\alpha < \displaystyle{\frac{1}{2-\frac{1}{M+1}}}$ such that $\displaystyle{\frac{1}{x^x(1-x)^{1-x}}}\geq 3^{\frac{1}{3}}$ where $x=1-\frac{1-\alpha}{M(2\alpha-1)+1-\alpha}$,\footnote{In particular, the result holds for any $\displaystyle{\frac{1}{2} < \frac{1}{2 - \frac{1}{7.35841\cdot M+1}}\leq \alpha < \frac{1}{2-\frac{1}{M+1}}}$.} \alg{ALG} runs in time $\displaystyle{O^*((\frac{1}{x^x(1-x)^{1-x}})^{vc})}$, returning a minimal vertex cover of weight at least $\alpha\cdot opt_w$. \alg{ALG} has a polynomial space complexity.
\end{theorem}

For example, for the smallest possible $\alpha$, the running time is bounded by $O^*(3^{\frac{vc}{3}}) < O^*(1.44225^{vc})$, and for any constant $\alpha < \frac{1}{2-\frac{1}{M+1}}$, there is a constant $\epsilon > 0$ such that the running time is bounded by $O^*((2-\epsilon)^{vc})$.

\alg{ALG} is based is on a mix of two bounded search tree-based procedures.\footnote{Information on the bounded search technique, which is perhaps the most well-known technique used to design parameterized algorithms \cite{newdfsbook}, is given in Section \ref{sec:approx}.} In particular, the branching vectors of one of these procedures are analyzed with respect to the size of {\em a minimum vertex cover of a minimum vertex cover} of $G$. Another interesting feature of this procedure is that once it reaches a leaf (of the search tree), it does not immediately return a result, but to obtain the desired approximation ratio, it first performs a computation that is, in a sense, an exhaustive search. We would like to note that in the design of our second procedure, we integrate rules that are part of the iterative compression-based algorithm for {\sc VC} by Peiselt \cite{vc2007}. Since \alg{ALG} can be used to solve {\sc MMVC}, in which case $M=1$, we immediately obtain the following corollary.

\begin{cor}\label{cor:unweightedApprox}
For any $\alpha < \frac{2}{3}$ such that $\frac{1}{x^x(1-x)^{1-x}}\geq 3^{\frac{1}{3}}$ where $x=2-\frac{1}{\alpha}$,\footnote{In particular, the result holds for any $0.53183\leq \alpha < \frac{2}{3}$.} \alg{ALG} runs in time $O^*((\frac{1}{x^x(1-x)^{1-x}})^{vc})$, returning a minimal vertex cover of size at least $\alpha\cdot opt$. \alg{ALG} has a polynomial space complexity.
\end{cor}

\section{Upper and Lower Bounds}\label{sec:simple}

In this section, we give upper and conditional lower bounds related to the parameterized complexity of {\sc WMMVC} and {\sc MMVC} with respect to $vc$ and $vc_w$. Recall, in this context of the next results, that $vc\leq vc_w$ and {\sc MMVC} is a special case of {\sc WMMVC}.

\begin{obs}
{\sc WMMVC} is solvable is time $O^*(2^{vc})$ and polynomial space.
\end{obs}

\begin{proof}
The algorithm is as follows. First, compute a minimum vertex cover $S$ of $G$ in time $O^*(\gamma^{vc})$ using polynomial space, and initialize $A$ to $S$. Then, for every subset $S'\subseteq S$, if $B=S'\cup (\bigcup_{v\in S\setminus S'}N(v))$ is a minimal vertex cover of weight larger than the weight of $A$, update $A$ to store $B$. Finally, return $A$.

Clearly, we return a minimal vertex cover. Now, let $A^*$ be an optimal solution. Consider the iteration where we examine $S'=A\cap S$. Then, since $A^*$ is a vertex cover, $B\subseteq A^*$. Suppose, by way of contradiction, that there exists $v\in A^*\setminus B$. The vertex $v$ does not belong to $S$ (since $A\cap S\subseteq B$). Moreover, it should have a neighbor outside $A^*$ (since $A^*$ is a {\em minimal} vertex cover). Thus, since $S$ is a vertex cover, $v$ has a neighbor in $S\setminus A^*$. This implies that $v\in \bigcup_{u\in S\setminus S'}N(u)$, which contradicts the assumption that $v\notin B$. We conclude that $B=A^*$. Therefore, the algorithm is correct, and it clearly has the desired time and space complexities.\qed
\end{proof}

Now, we observe that even in a restricted setting, the algorithm above is essentially optimal under the SETH.

\begin{lemma}\label{lemma:lower1}
For any constant $\epsilon>0$, {\sc WMMVC} in bipartite graphs, where $w(v)\in\{1,1+\frac{1}{|V|}\}$ for all $v\in V$, cannot be solved in time $O^*((2-\epsilon)^{vc_w})$ unless the SETH fails.
\end{lemma}

\begin{proof}
Fix $\epsilon>0$. Suppose, by way of contradiction, that there exists an algorithm, \alg{A}, that solves {\sc WMMVC} in the restricted setting in time $O^*((2-\epsilon)^{vc_w})$. We aim to show that this implies that there exists an algorithm that solves the {\sc Hitting Set (HS)} problem in time $O^*((2-\epsilon)^n)$, which contradicts the SETH \cite{cnfSAT}. In {\sc HS}, we are given an $n$-element set $U$, along with family of subsets of $U$, ${\cal F}=\{F_1,F_2,\ldots,F_m\}$, and the goal is to find the minimum size of a subset $U'\subseteq U$ that is a hitting set (i.e., $U'$ contains at least one element from every set in $\cal F$). 

We next construct an instance $(G=(V,E), w: V\rightarrow\mathbb{R}^{\geq 1})$ of {\sc WMMVC} in the restricted setting:
\begin{itemize}
\item $R_1=\{r_u: u\in U\}$, and $R_2=\{r^c_i: c\in\{1,\ldots,n+1\}, i\in\{1,\ldots,m\}\}$.
\item $L=\{l_u: u\in U\}$, and $R=R_1\cup R_2$.
\item $V=L\cup R$.
\item $E=\{\{l_u,r_u\}: u\in U\}\cup\{\{l_u,r^c_i\}: u\in F_i, i\in\{1,\ldots,m\}, c\in\{1,\ldots,n+1\}\}$.
\item $\forall v\in L: w(v)=1+\frac{1}{|V|}$.
\item $\forall v\in R: w(v)=1$.
\end{itemize} 

\begin{figure}[!t]\centering
\frame{\includegraphics[scale=0.75]{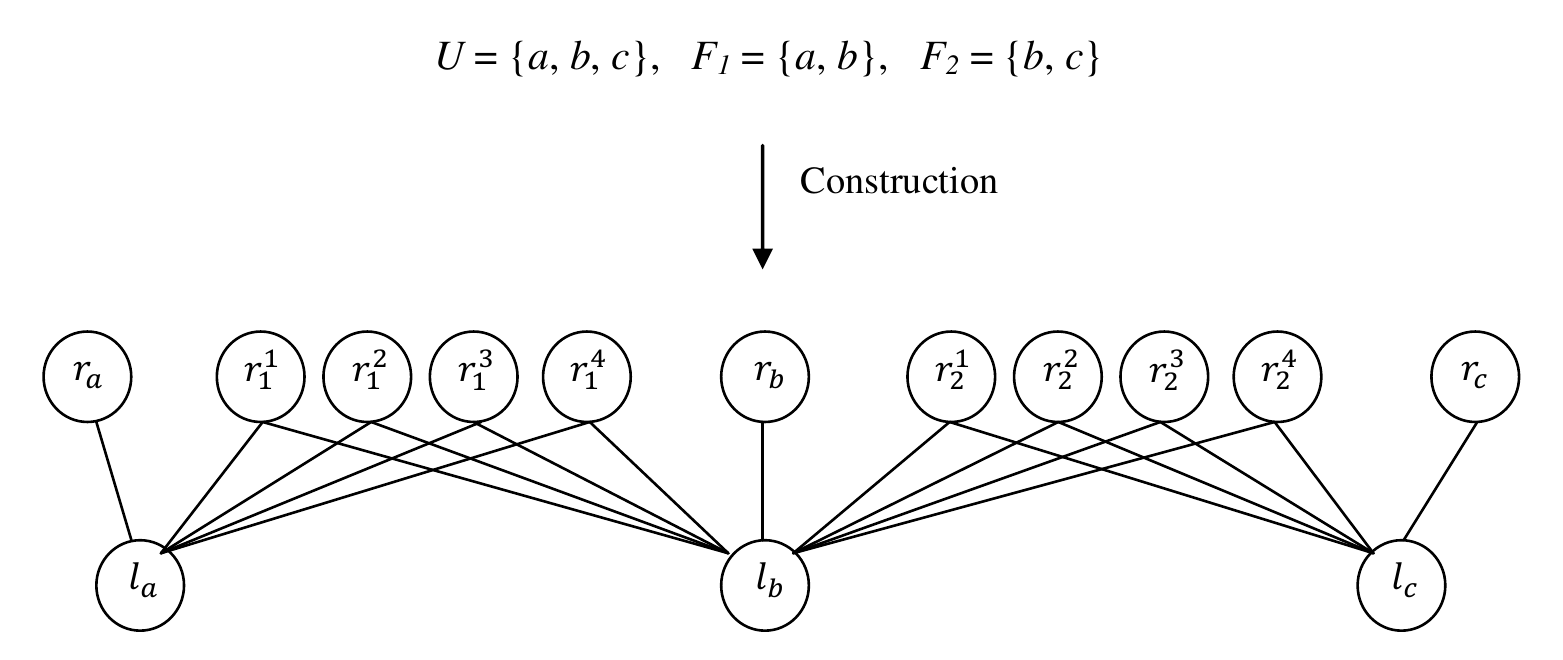}}
\caption{The construction in the proof of Lemma \ref{lemma:lower1}.}
\label{fig:reduction}
\end{figure}

An illustrated example is given in Fig.~\ref{fig:reduction}. It is enough to show that (1) $vc\leq n$, and (2) the solution for $(U,{\cal F})$ is $q$ {\em iff} the solution for $(G,w)$ is $\frac{n-q}{|V|} + |R|$. 
Indeed, this implies that we can solve {\sc HS} by constructing the above instance in polynomial time, running \alg{A} in time $O^*((2-\epsilon)^{vc})$ (since $vc\leq n$), obtaining an answer of the form $\frac{n-q}{|V|} + |R|$, and returning $q$.

First, observe that $L$ is a minimal vertex cover of $G$: it is a vertex cover, since every edge has exactly one endpoint in $L$ and one endpoint in $R$, and it is minimal, since every vertex in $L$ has an edge that connects is to at least one vertex in $R$. Therefore, $vc\leq |L|=n$.

Now, we turn to prove the second item. For the first direction, let $q$ be the solution to $(U,{\cal F})$, and let $U'$ be a corresponding hitting set of size $q$. Consider the vertex set $S=\{l_u: u\in U\setminus U'\}\cup \{r_u: u\in U'\}\cup R_2$. By the definition of $w$, the weight of $S$ is $(1+\frac{1}{|V|})|U\setminus U'| + |U'| + |R_2| = \frac{1}{|V|}|U\setminus U'| + |R_1| + |R_2| = \frac{n-q}{|V|} + |R|$. The set $S$ is a vertex cover: since $R_2\subseteq S$, every edge in $G$ that does not have an endpoint in $R_2$ is of the form $\{l_u,r_u\}$, and for every $u\in U$ either $l_u\in S$ (if $u\notin U'$) or $r_u\in S$ (if $u\in U'$). Moreover, $S$ is a {\em minimal} vertex cover. Indeed, we cannot remove any vertex $l_u\in S\cap L$ or $r_u\in S\cap R_1$ and still have a vertex cover, since then the edge of the form $\{l_u,r_u\}$ is not covered. Also, we cannot remove any vertex $r^c_i\in R_2$, since there is a vertex $l_u\notin S$ such that $\{l_u,r^c_i\}\in E$ (to see this, observe that because $U'$ is a hitting set, there is a vertex $u\in U'\cap F_i$, which corresponds to the required vertex $l_u$).

For the second direction, let $p$ be the solution to $(G,w)$, and let $S$ be a corresponding minimal vertex cover of weight $p$. Clearly, $p\geq w(R) = n + m(n+1)$, since $R$ is a minimal vertex cover of $G$. Observe that for all $u\in U$, by the definition of $G$ and since $S$ is a minimal vertex cover, exactly one among the vertices $l_u$ and $r_u$ is in $S$. Suppose that there exists $r^c_i\in R_2\setminus S$. Then, for all $u\in F_i$, we have that $l_u\in S$ (by the definition of $G$ and since $S$ is a vertex cover), which implies that for all $c\in\{1,\ldots,n+1\}$, we have that $r^c_i\notin S$ (since $S$ is a {\em minimal} vertex cover). Thus, $p = w(S\cap (L\cup R_1)) + |S\cap R_2|\leq n(1+\frac{1}{|V|}) + (m-1)(n+1) < n + m(n+1)$, which is a contradiction. Therefore, $R_2\subseteq S$.

Denote $U'=\{u: r_u\in S\cap R_1\}$. By the above discussion and the definition of $w$, $|S| = \frac{|V\setminus R_2|}{2} + |R_2| = |R_1|+|R_2|=|R|$, and $p-|S|=\frac{1}{|V|}|S\cap L|=\frac{1}{|V|}(n-|S\cap R_1|)$. Denoting $|U'|=q$, we have that $p=\frac{n-q}{|V|} + |R|$. Thus, it remains to show that $U'$ is a hitting set. Suppose, by way of contradiction, that $U'$ is not a hitting set. Thus, there exists $F_i\in{\cal F}$ such that for all $u\in F_i$, we have that $u\notin U'$. By the definition of $U'$, this implies that for all $u\in F_i$, we have that $l_u\in S$. Thus, $N(r^1_i)\subseteq S$, while $r^1_i\in S$ (since we have shown that $R_2\subseteq S$), which is a contradiction to the fact that $S$ is a {\em minimal} vertex cover.\qed
\end{proof}

Next, we also give a conditional lower bound for {\sc MMVC}. The proof is quite similar to the one above, and is thus relegated to Appendix \ref{app:simple}. Informally, the idea is to modify the above proof by adding a copy $l_u'$ of every vertex $l_u\in L$, which is only attached to its ``mirror vertex'', $r_u$, in $R_1$. While previously the weights $1+\frac{1}{|V|}$ encouraged us to choose many vertices from $L$, now the copies encourage us to choose these vertices.

\begin{lemma}\label{lemma:lower2}
For any constant $\epsilon>0$, {\sc MMVC} in bipartite graphs cannot be solved in time $O^*((2-\epsilon)^{\frac{vc}{2}})$ unless the SETH fails.
\end{lemma}

\section{A Parameterized Approximation Algorithm}\label{sec:approx}

In this section we develop \alg{ALG} (see Theorem \ref{theorem:approx}). When referring to $\alpha$, suppose that it is chosen as required in Theorem \ref{theorem:approx}, and define $x$ accordingly. The algorithm is a mix of two bounded search tree-based procedures, \alg{ProcedureA} and \alg{ProcedureB}, which are developed in the following subsections.\footnote{\alg{ProcedureA} could also be developed without using recursion; however, relying on the bounded search tree technique simplifies the presentation, emphasizing the parts similar between \alg{ProcedureA} and \alg{ProcedureB}.} For these procedures, we will prove the following lemmas.

\begin{lemma}\label{lemma:proA}
Let $U$ be a minimum(-size) vertex cover of $G$, and let $U'$ be a minimum(-size) vertex cover of $G[U]$. Moreover, suppose that $|U'|\leq \frac{vc}{2}$. 

Then, \alg{ProcedureA}$(G,w,\alpha,U,U',\emptyset,\emptyset)$ runs in time $O^*((\frac{1}{x^x(1-x)^{1-x}})^{vc})$, using polynomial space and returning a minimal vertex cover of weight at least $\alpha\cdot opt_w$.
\end{lemma}

\begin{lemma}\label{lemma:proB}
Let $U$ be a minimum(-size) vertex cover of $G$ such that the size of a minimum(-size) vertex cover of $G[U]$ is larger than $\frac{vc}{2}$. 

Then, \alg{ProcedureB}$(G,w,U,\emptyset,\emptyset)$ runs in time $O^*(3^{\frac{vc}{3}})$, using polynomial space and returning a minimal vertex cover of weight at least $\displaystyle{\frac{1}{2-\frac{1}{M+1}}}\cdot opt_w$.
\end{lemma}

Having these procedures, we give the pseudocode of \alg{ALG} below. The algorithm computes a minimum vertex cover $U'$ of a minimum vertex cover $U$, solving the given instance by calling either \alg{ProcedureA} (if $|U'|\leq\frac{vc}{2}$) or \alg{ProcedureB} (if $|U'|>\frac{vc}{2}$).

\begin{algorithm}[!ht]
\caption{\alg{ALG}($G=(V,E),w: V\rightarrow \mathbb{R}^{\geq 1},\alpha$)}
\begin{algorithmic}[1]
\STATE Compute a minimum vertex cover $U$ of $G$ in time $O^*(\gamma^{vc})$ and polynomial space.
\STATE Compute a minimum vertex cover $U'$ of $G[U]$ in time $O^*(\gamma^{vc})$ and polynomial~space.
\IF{$|U'|\leq \frac{vc}{2}$}
	\STATE Return \alg{ProcedureA}$(G,w,\alpha,U,U',\emptyset,\emptyset)$.
\ELSE
	\STATE Return \alg{ProcedureB}$(G,w,U,\emptyset,\emptyset)$.
\ENDIF
\end{algorithmic}
\end{algorithm}

Now, we turn to prove Theorem \ref{theorem:approx}.

\begin{proof}
The correctness of the approximation ratio immediately follows from Lemmas \ref{lemma:proA} and \ref{lemma:proB}. Observe that the computation of $U'$ can indeed be performed in time $O^*(\gamma^{vc})$ since $U'\subseteq U$, and therefore $|U'|\leq|U|\leq vc$. Moreover, since $\gamma<1.274$, by Lemmas \ref{lemma:proA} and \ref{lemma:proB}, the running time is bounded by $O^*(\max\{(\frac{1}{x^x(1-x)^{1-x}})^{vc},3^{\frac{vc}{3}}\})$. Since $\alpha$ is chosen such that $(\frac{1}{x^x(1-x)^{1-x}})^{vc}\geq 3^{\frac{vc}{3}}$, the above running time is bounded by $O^*((\frac{1}{x^x(1-x)^{1-x}})^{vc})$.\qed
\end{proof}

In the rest of this section, we give necessary information on the bounded search tree technque (Section \ref{sec:bounded}), after which we develop \alg{ProcedureA} (Section \ref{sec:proA}) and \alg{ProcedureB} (Section \ref{sec:proB}).

\subsection{The Bounded Search Tree Technique}\label{sec:bounded}

Bounded search trees form a fundamental technique in the design of recursive parameterized algorithms (see \cite{newdfsbook}). Roughly speaking, in applying this technique, one defines a list of rules of the form \alg{Rule X. [condition] action}, where \alg{X} is the number of the rule in the list. At each recursive call (i.e., a node in the search tree), the algorithm performs the action of the first rule whose condition is satisfied. If by performing an action, the algorithm recursively calls itself at least twice, the rule is a {\em branching rule}, and otherwise it is a {\em reduction rule}. We only consider actions that increase neither the parameter nor the size of the instance, and decrease at least one of them. Observe that, at any given time, we only store the path from the current node to the root of the search tree (rather than the entire tree).

The running time of the algorithm can be bounded as follows. Suppose that the algorithm executes a branching rule where it recursively calls itself $\ell$ times, such that in the $i^\mathrm{th}$ call, the current value of the parameter decreases by $b_i$. Then, $(b_1,b_2,\ldots,b_{\ell})$ is called the {\em branching vector} of this rule. We say that $\beta$ is the {\em root} of $(b_1,b_2,\ldots,b_{\ell})$ if it is the (unique) positive real root of $x^{b^*} = x^{b^*-b_1} + x^{b^*-b_2} + \ldots + x^{b^*-b_{\ell}}$, where $b^* = \max\{b_1,b_2,\ldots,b_{\ell}\}$. If $b>0$ is the initial value of the parameter, and the algorithm (a) returns a result when (or before) the parameter is negative, (b) only executes branching rules whose roots are bounded by a constant $c$, and (c) only executes rules associated with actions performed in polynomial time, then its running time is bounded by $O^*(c^b)$.

In some of the leaves of a search tree corresponding to our first procedure, we execute rules associated with actions that are not performed in polynomial time (as required in the condition (c) above).
We show that for every such leaf $\ell$ in the search tree, we execute an action that can be performed in time $O^*(g(\ell))$ for some function $g$. Then, letting $L$ denote the set of leaves in the search tree whose actions are not performed in polynomial time, we have that the running time of the algorithm is bounded by $O^*(c^b+\sum_{\ell\in L}g(\ell))$.

\subsection{ProcedureA: The Proof of Lemma \ref{lemma:proA}}\label{sec:proA}

The procedure \alg{ProcedureA} is based on the bounded search tree technique. Each call is of the form \alg{ProcedureA}$(G,w,\alpha,U,U',I,O)$, where $G,w,\alpha,U$ and $U'$ always remain the parameters with whom the procedure was called by \alg{ALG}, while $I$ and $O$ are disjoint subsets of $U$ to which \alg{ProcedureA} adds elements as the execution progresses (initially, $I=O=\emptyset$). Roughly speaking, the sets $I$ and $O$ indicate that currently we are only interested in examining minimal vertex covers that contain all of the vertices in $I$ and none of the vertices in $O$. Formally, we prove the following result.

\begin{lemma}\label{lemma:presProA}
\alg{ProcedureA} returns a minimal vertex cover $S$ that satisfies the following condition:
\begin{itemize}
\item If there is a minimal vertex cover $S^*$ of weight $opt_w$ such that $I\subseteq S^*$ and $O\cap S^*=\emptyset$, then the weight of $S$ is at least $\alpha\cdot opt_w$.
\end{itemize}
Moreover, each leaf, associated with an instance $(G,w,\alpha,U,U',I',O')$, corresponds to a unique pair $(I',O')$, and its action can be performed in polynomial time if $I'\cup O'\neq U'$, and in time $O^*(|\{\widetilde{U}\subseteq U: I'\subseteq \widetilde{U}, O'\cap \widetilde{U}=\emptyset, |\widetilde{U}|\geq (1-x)vc\}|)$ otherwise.
\end{lemma}

Let $vc'=|U'|$. To ensure that \alg{ProcedureA} runs in time $O^*((\frac{1}{x^x(1-x)^{1-x}})^{vc})$, we propose the following measure:

\smallskip
{\noindent{\bf Measure:} $\displaystyle{vc'-|U'\cap(I\cup O)|}$.}
\smallskip

Next, we present each rule within a call \alg{ProcedureA}$(G,w,\alpha,U,U',I,O)$. After presenting a rule, we argue its correctness (see Lemma \ref{lemma:presProA}). Since initially $I=O=\emptyset$, we thus have that \alg{ProcedureA} guarantees the desired approximation ratio. For each branching rule, we also give the root of the corresponding branching vector (with respect to the measure above). We ensure that (1) the largest root we shall get is 2, (2) the procedure stops calling itself recursively, at the latest, once the measure drops to 0, and (3) actions not associated with leaves can be performed in polynomial time.
Observe that initially the measure is $vc'$. Thus, as explained in Section \ref{sec:bounded}, the running time of \alg{ProcedureA} is bounded by $\displaystyle{O^*(2^{vc'} + \sum_{(I',O')\in{\cal P}}|\{\widetilde{U}\subseteq U: I'\subseteq \widetilde{U}, O'\cap \widetilde{U}=\emptyset, |\widetilde{U}|\geq (1-x)vc\}|)}$, where $\cal P$ is the set of all partitions of $U'$ into two sets.
This running time is bounded by $\displaystyle{O^*(2^{vc'} + |\{\widetilde{U}\subseteq U: |\widetilde{U}|\geq (1-x)vc\}|)=O^*(2^{vc'} + \max_{i=(1-x)vc}^{vc}{vc \choose i})}$.
Since $x < \frac{1}{2}$ (by the definition of $x$ and since $\alpha<\displaystyle{\frac{1}{2-\frac{1}{M+1}}}$), $vc'\leq\frac{vc}{2}$ and $\frac{1}{x^x(1-x)^{1-x}}\geq 3^{\frac{1}{3}}> 2^{\frac{1}{2}}$, this running time is further bounded by $\displaystyle{O^*((\frac{1}{x^x(1-x)^{1-x}})^{vc})}$.
Thus, we also have the desired bound for the running time, concluding the correctness of Lemma \ref{lemma:presProA}.


\begin{reducerule}\label{red:stopF1}{\normalfont
[There is $v\in O$ such that $N(v)\cap O\neq\emptyset$]
Return $U$.
}\end{reducerule}

{\noindent In this case there is no vertex cover that does not contain any vertex from $O$, and therefore it is possible to return an arbitrary minimal vertex cover. The action can clearly be performed in polynomial time.}

\begin{reducerule}\label{red:stopF2}{\normalfont
[There is $v\in X$ such that $N(v)\subseteq X$, where $X=I\cup (\bigcup_{u\in O}N(u))$]
Return $U$.
}\end{reducerule}

{\noindent Observe that any vertex cover that does not contain any vertex from $O$, must contain all the neighbors of the vertices in $O$. Thus, any vertex cover that contains all the vertices in $I$ and none of the vertices in $O$, also contains the vertex $v$ and all of its neighbors; therefore, it is not a {\em minimal} vertex cover. Thus, it is possible to return an arbitrary minimal vertex cover. The action can clearly be performed in polynomial time.}

\begin{reducerule}\label{red:stopT}{\normalfont
[$U'=I\cup O$] Perform the following computation.
\begin{enumerate}
\item Let $A=I\cup (\bigcup_{v\in O}N(v)\cap U)$. As long as there is a vertex $v\in A$ such that $N(v)\cap U\subseteq A$, choose such a vertex (arbitrarily) and remove it from $A$. Let $A'$ be the set obtained at the end of this process.
\smallskip
\item Let $\widetilde{A} = A'\cup (\bigcup_{v\in U\setminus A'}N(v)\setminus U)$.
\smallskip
\item Denote ${\cal F}=\{\widetilde{U}\subseteq U: I\subseteq\widetilde{U}, O\cap\widetilde{U}=\emptyset, |\widetilde{U}|\geq (1-x)vc\}$.\footnote{It is not necessary to explicitly store $\cal F$---we only need to iterate over it; therefore, by the pseudocode, it is clear that the space complexity of the action is polynomial.}
\smallskip
\item Initialize $B=U$.
\smallskip
\item For all $F\in{\cal F}$:
	\begin{enumerate}
	\item Let $B_F = F\cup (\bigcup_{v\in U\setminus F}N(v))$
	\item If $B_F$ is a {\em minimal} vertex cover:
		\begin{itemize}
		\item If the weight of $B_F$ is larger than the weight of $B$: Replace $B$ by $B_F$.
		\end{itemize}
	\end{enumerate}
\smallskip	
\item Return the set of maximum weight among $\widetilde{A}$ and $B$.
\end{enumerate}
}\end{reducerule}

{\noindent First, observe that this rule ensure that, at the latest, the procedure stops calling itself recursively once the measure drops to 0. Furthemore, by the pseudocode, the action can be performed in time $O^*(|{\cal F}|)$, which by the definition of ${\cal F}$, is the desired time.
It remains to prove that Lemma \ref{lemma:proA} is correct.}

We begin by considering the set $A$. Since $U'=I\cup O$ is a vertex cover of $G[U]$ and the previous rules were not applied, we have that $A$ is vertex cover of $G[U]$ such that $I\subseteq A$ and $O\cap A=\emptyset$. Thus, by its definition, $A'$ is a {\em minimal} vertex cover of $G[U]$ such that $A'\subseteq A$. Thus, since $U$ is a vertex cover, every edge either has both endpoints in $U$, in which case it has an endpoint in $A'$, or it has exactly one endpoint in $U$ and exactly one endpoint in $V\setminus U$, in which case it has an endpoint that is a vertex in $A'$ or a neighbor in $V\setminus U$ of a vertex in $U\setminus A'$. Therefore, $\widetilde{A}$ is a vertex cover. Moreover, every vertex in $A'$ has a neighbor in $U\setminus A'$ (by the minimality of $A'$) and every vertex in $(\bigcup_{v\in U\setminus A'}N(v)\setminus U)$, by definition, has a neighbor in $U\setminus A'$. Thus, we overall have that $\widetilde{A}$ is a minimal vertex cover such that $\widetilde{A}\cap U\subseteq A$.

Since $\widetilde{A}$ is a minimal vertex cover, by the pseudocode, we return a weight, $W$, of a minimal vertex cover. Assume that there is a minimal vertex cover $S^*$ of weight $opt_w$ such that $I\subseteq S^*$ and $O\cap S^*=\emptyset$. This implies that $A\subseteq S^*\cap U$. Now, to prove Lemma \ref{lemma:proA}, it is sufficient to show that $W\geq\alpha\cdot opt_w$. Denote $F^*=S^*\cap U$. Since $S^*$ is a minimal vertex cover, we have that $S^*=F^*\cup(\bigcup_{v\in U\setminus F^*}N(v))$. If $S^*$ contains at least $(1-x)vc$ elements from $U$, there is an iteration where we examine $F=F^*$, in which case $B_{F^*}=S^*$, and therefore we return $opt_w$.
Thus, we next suppose that $|S^*\cap U| < (1-x)vc$. 
Since $B$ is initially $U$, to prove Lemma \ref{lemma:proA}, it is now sufficient to show that $\max\{w(\widetilde{A}),w(U)\}\geq\alpha\cdot w(S^*)$.

Since $S^*$ is a vertex cover such that $I\subseteq S^*$ and $O\cap S^*=\emptyset$, we have that $\widetilde{A}\cap U\subseteq A\subseteq F^*$. By the definition of $\widetilde{A}$ and since $S^*$ is a {\em minimal} vertex cover, this implies that $S^*\setminus U\subseteq \widetilde{A}\setminus U$. Thus, overall we have that

\[\begin{array}{ll}

\medskip
\displaystyle{\frac{\max\{w(U),w(\widetilde{A})\}}{w(S^*)}} & = \displaystyle{\frac{\max\{w(U),w(\widetilde{A})\}}{w(S^*\setminus U) + w(S^*\cap U)}}
\geq \displaystyle{\frac{\max\{w(U), w(\widetilde{A})\}}{w(\widetilde{A}\setminus U) + w(S^*\cap U)}}\\

\medskip
& = \displaystyle{\frac{\max\{w(U), w(\widetilde{A})\}}{w(\widetilde{A}) + w(S^*\cap U)-w(\widetilde{A}\cap U)}}\\

\medskip
& \geq \displaystyle{\frac{w(U)}{w(U) + w(S^*\cap U)}} = \displaystyle{\frac{1}{2-\frac{w(U\setminus S^*)}{w(U)}}} \geq \displaystyle{\frac{1}{2-\frac{|U\setminus S^*|}{w(S^*\cap U) + |U\setminus S^*|}}}\\

\medskip
& \geq \displaystyle{\frac{1}{2-\frac{x\cdot vc}{M\cdot(1-x)\cdot vc + x\cdot vc}}} = \displaystyle{\frac{1}{2-\frac{x}{M-(M-1)x}}}
\end{array}\]

Since $x=1-\frac{1-\alpha}{M(2\alpha-1)+1-\alpha}$, we have that the expression above is equal to $\alpha$.

\begin{branchrule}{\normalfont
Let $v$ be a vertex in $U'\setminus(I\cup O)$. Return the set of maximum weight among $A$ and $B$, computed in the following branches.
\begin{enumerate}
\item $A\Leftarrow$\alg{ProcedureA}$(G,w,\alpha,U,U',I\cup\{v\},O)$.
\item $B\Leftarrow$\alg{ProcedureA}$(G,w,\alpha,U,U',I,O\cup\{v\})$.
\end{enumerate}
}\end{branchrule}

{\noindent The correctness of Lemma \ref{lemma:presProA} is preserved, since every vertex cover either contains $v$ (an option examined in the first branch) or does not contain $v$ (an option examined in the second branch). Moreover, it is clear that the action can be performed in polynomial time and that the branching vector is (1,1), whose root is 2.}

\subsection{ProcedureB: The Proof of Lemma \ref{lemma:proB}}\label{sec:proB}

This procedure is based on combining an appropriate application of the ideas used by the previous procedure (considering the fact that now the size of any vertex cover of $G[U]$ is larger than $\frac{vc}{2}$) with rules from the algorithm for {\sc VC} by Peiselt \cite{vc2007}. Due to lack of space, the details are given in Appendix \ref{app:proB}.

\bibliographystyle{splncs03}
\bibliography{References}

\newpage

\appendix

\section{Proof of Lemma \ref{lemma:lower2}}\label{app:simple}

Fix $\epsilon>0$. Suppose, by way of contradiction, that there exists an algorithm, \alg{A}, that solves {\sc MMVC} in the restricted setting in time $O^*((2-\epsilon)^{\frac{vc}{2}})$. We aim to show that this implies that there exists an algorithm that solves {\sc HS} in time $O^*((2-\epsilon)^n)$, which contradicts the SETH \cite{cnfSAT}. To this end, we construct an graph $G=(V,E)$ that defines an instance of {\sc MMVC} in the restricted setting:
\begin{itemize}
\item $R_1=\{r_u: u\in U\}$, and $R_2=\{r^c_i: c\in\{1,\ldots,n+1\}, i\in\{1,\ldots,m\}\}$.
\item $L=\{l^c_u: u\in U, c\in\{1,2\}\}$, and $R=R_1\cup R_2$.
\item $V=L\cup R$.
\item $E=\{\{l^c_u,r_u\}: u\in U, c\in \{1,2\}\}\cup\{\{l^1_u,r^c_i\}: u\in F_i, i\in\{1,\ldots,m\}, c\in\{1,\ldots,n+1\}\}$.
\end{itemize} 

It is enough to show that (1) $vc\leq 2n$, and (2) the solution for $(U,{\cal F})$ is $q$ {\em iff} the solution for $(G,w)$ is $(n-q) + |R|$. 
Indeed, this implies that we can solve {\sc HS} by constructing the above instance in polynomial time, running \alg{A} in time $O^*((2-\epsilon)^n)$ (since $vc\leq 2n$), obtaining an answer of the form $(n-q) + |R|$, and returning $q$.

First, observe that $L$ is a minimal vertex cover of $G$: it is a vertex cover, since every edge has exactly one endpoint in $L$ and one endpoint in $R$, and it is minimal, since every vertex in $L$ has an edge that connects is to at least one vertex in $R$. Therefore, $vc\leq |L|=2n$.

Now, we turn to prove the second item. For the first direction, let $q$ be the solution to $(U,{\cal F})$, and let $U'$ be a corresponding hitting set of size $q$. Consider the vertex set $S=\{l^c_u: u\in U\setminus U', c\in\{1,2\}\}\cup \{r_u: u\in U'\}\cup R_2$. Observe that $|S| = 2|U\setminus U'| + |U'| + |R_2| = |U\setminus U'| + |R_1| + |R_2| = (n-q) + |R|$. The set $S$ is a vertex cover: since $R_2\subseteq S$, every edge in $G$ that does not have an endpoint in $R_2$ is of the form $\{l^c_u,r_u\}$, and for every $u\in U$ either [$l^1_u,l^2_u\in S$ (if $u\notin U'$)] or [$r_u\in S$ (if $u\in U'$)]. Moreover, $S$ is a {\em minimal} vertex cover. Indeed, we cannot remove any vertex $l_u\in S\cap L$ or $r_u\in S\cap R_1)$ and still have a vertex cover, since then the edge of the form $\{l^c_u,r_u\}$ is not covered. Also, we cannot remove any vertex $r^c_i\in R_2$, since there is a vertex $l^1_u\notin S$ such that $\{l^1_u,r^c_i\}\in E$ (to see this, observe that because $U'$ is a hitting set, there is a vertex $u\in U'\cap F_i$, which corresponds to the required vertex $l^1_u$).

For the second direction, let $p$ be the solution to the instance defined by $G$, and let $S$ be a corresponding minimal vertex cover of size $p$. Clearly, $p\geq |R| = n + m(n+1)$, since $R$ is a minimal vertex cover of $G$. Observe that for all $u\in U$, by the definition of $G$ and since $S$ is a minimal vertex cover, we can assume WLOG that either [$l^1_u,l^2_u\in S$ and $r_u\notin S$] or [$l^1_u,l^2_u\notin S$ and $r_u\in S$]. Indeed, to see this, note that if $l^1_u,r_u\in S$, then $l^2_u\notin S$, and thus we can replace (in $S$) $r_u$ by $l^2_u$ and still have a solution. Suppose that there exists $r^c_i\in R_2\setminus S$. Then, for all $u\in F_i$, we have that $l^1_u\in S$ (by the definition of $G$ and since $S$ is a vertex cover), which implies that for all $c\in\{1,\ldots,n+1\}$, we have that $r^c_i\notin S$ (since $S$ is a {\em minimal} vertex cover). Thus, $p = |S\cap (L\cup R_1)| + |S\cap R_2|\leq 2n + (m-1)(n+1) < n + m(n+1)$, which is a contradiction. Therefore, $R_2\subseteq S$.

Denote $U'=\{u: r_u\in S\cap R_1\}$. By the above discussion, $p = |S\cap L| + |S\cap R_1|  + |R_2| = \frac{1}{2}|S\cap L|+|R_1|+|R_2|=(n-|S\cap R_1|)+|R|$. Denoting $|U'|=q$, we have that $p= (n-q) + |R|$. Thus, it remains to show that $U'$ is a hitting set. Suppose, by way of contradiction, that $U'$ is not a hitting set. Thus, there exists $F_i\in{\cal F}$ such that for all $u\in F_i$, we have that $u\notin U'$. By the definition of $U'$, this implies that for all $u\in F_i$, we have that $l^1_u\in S$. Thus, $N(r^1_i)\subseteq S$, while $r^1_i\in S$ (since we have shown that $R_2\subseteq S$), which is a contradiction to the fact that $S$ is a {\em minimal} vertex cover.\qed

\section{ProcedureB: The Proof of Lemma \ref{lemma:proB} (Cont.)}\label{app:proB}

The procedure \alg{ProcedureB} is based on the bounded search tree technique. Each call is of the form \alg{ProcedureB}$(G,w,U,I,O)$, where $G,w$ and $U$ always remain the parameters with whom the procedure was called by \alg{ALG}, while $I$ and $O$ are disjoint subsets of $U$ to which \alg{ProcedureB} adds elements as the execution progresses (initially, $I=O=\emptyset$). As in the case of \alg{ProcedureA}, the sets $I$ and $O$ indicate that currently we are only interested in examining minimal vertex cover that contains all of the vertices in $I$ and none of the vertices in $O$. Formally, we prove the following result.

\begin{lemma}\label{lemma:presProB}
\alg{ProcedureB} returns a minimal vertex cover $S$ that satisfies the following condition:
\begin{itemize}
\item If there is a minimal vertex cover $S^*$ of weight $opt_w$ such that $I\subseteq S^*$ and $O\cap S^*=\emptyset$, then the weight of $S$ is at least $\displaystyle{\frac{1}{2 - \frac{1}{M+1}}}\cdot opt_w$.
\end{itemize}
\end{lemma}

To ensure that \alg{ProcedureB} runs in time $O^*(3^{\frac{vc}{3}})$, we use the measure below:

\smallskip
{\noindent{\bf Measure:} $\displaystyle{vc-|U\cap(I\cup O)|-|S(I\cup O)|}$, where $S(I\cup O)$ contains the vertices in $U\setminus(I\cup O)$ that do not have a neighbor in $U\setminus (I\cup O)$.}
\smallskip

Next, we present each rule within a call \alg{ProcedureB}$(G,w,U,I,O)$. After presenting a rule, we argue its correctness (see Lemma \ref{lemma:presProB}). Since initially $I=O=\emptyset$, we thus have that \alg{ProcedureB} guarantees the desired approximation ratio. For each branching rule, we also give the root of the corresponding branching vector (with respect to the measure above). We ensure that (1) the largest root we shall get is at most $3^{\frac{1}{3}}$, (2) the procedure stops calling itself recursively, at the latest, once the measure drops to 0, and (3) the procedure only executes rules whose actions can be performed in polynomial time. Observe that initially the measure is $vc$. Thus, as explained in Section \ref{sec:bounded}, the running time of \alg{ProcedureB} is bounded by $O^*(3^{\frac{vc}{3}})$.

\setcounter{reducerule}{0}
\begin{reducerule}{\normalfont
[There is $v\in O$ such that $N(v)\cap O\neq\emptyset$]
Return $U$.
}\end{reducerule}

{\noindent Follow the explanation given for Rule \ref{red:stopF1} of \alg{ProcedureA}.}

\begin{reducerule}{\normalfont
[There is $v\in X$ such that $N(v)\subseteq X$, where $X=I\cup (\bigcup_{u\in O}N(u))$]
Return $U$.
}\end{reducerule}

{\noindent Follow the explanation given for Rule \ref{red:stopF2} of \alg{ProcedureA}.}

\begin{reducerule}{\normalfont
[$U=I\cup O\cup S(I\cup O)$] Perform the following computation.
\begin{enumerate}
\item Let $A=I\cup (\bigcup_{v\in O}N(v)\cap U)$. As long as there is a vertex $v\in A$ such that $N(v)\cap U\subseteq A$, choose such a vertex (arbitrarily) and remove it from $A$. Let $A'$ be the set obtained at the end of this process.
\smallskip
\item Let $\widetilde{A} = A'\cup (\bigcup_{v\in U\setminus A'}N(v)\setminus U)$.
\smallskip
\item Return the set of maximum weight among $\widetilde{A}$ and $U$.
\end{enumerate}
}\end{reducerule}

{\noindent Observe that this rule ensure that, at the latest, the procedure stops calling itself recursively once the measure drops to 0. We next prove that Lemma \ref{lemma:proB} is correct. In a manner similar to the explanation following Rule \ref{red:stopT}, we have that $\widetilde{A}$ is a minimal vertex cover such that $\widetilde{A}\cap U\subseteq A$. In particular, $\widetilde{A}\cap U$ is a minimal vertex cover of $G[U]$, and therefore its size is larger than $\frac{vc}{2}$ (otherwise \alg{ALG} would have called \alg{ProcedureA}). Thus, by the pseudocode, we return a weight of a minimal vertex cover. Assume that there is a minimal vertex cover $S^*$ of weight $opt_w$ such that $I\subseteq S^*$ and $O\cap S^*=\emptyset$. This implies that $A\subseteq S^*\cap U$. Now, it is sufficient to show that $\max\{w(\widetilde{A}),w(U)\}\geq\alpha\cdot w(S^*)$. As in the explanation following Rule \ref{red:stopT}, $S^*\setminus U\subseteq \widetilde{A}\setminus U$. Thus, we have that}

\[\begin{array}{ll}

\medskip
\displaystyle{\frac{\max\{w(U),w(\widetilde{A})\}}{w(S^*)}} & = \displaystyle{\frac{\max\{w(U),w(\widetilde{A})\}}{w(S^*\setminus U) + w(S^*\cap U)}}
\geq \displaystyle{\frac{\max\{w(U), w(\widetilde{A})\}}{w(\widetilde{A}\setminus U) + w(S^*\cap U)}}\\

\medskip
& = \displaystyle{\frac{\max\{w(U), w(\widetilde{A})\}}{w(\widetilde{A}) + w(S^*\cap U)-w(\widetilde{A}\cap U)}}\\

\medskip
& \geq \displaystyle{\frac{w(U)}{w(U) + w((S^*\setminus\widetilde{A})\cap U)}} = \displaystyle{\frac{w(U)}{2w(U) - w(U\setminus(S^*\setminus\widetilde{A}))}}\\

\medskip
& = \displaystyle{\frac{1}{2 - \frac{w(U\setminus(S^*\setminus\widetilde{A}))}{w(U)}}} \geq \displaystyle{\frac{1}{2 - \frac{w(U\cap\widetilde{A})}{w(U)}}}\\

& \geq \displaystyle{\frac{1}{2 - \frac{\frac{vc}{2}}{M\cdot\frac{vc}{2}+\frac{vc}{2}}}} = \displaystyle{\frac{1}{2 - \frac{1}{M+1}}}
\end{array}\]

Next, we denote $\widehat{U}= U\setminus(I\cup O)$. In the remaining (branching) rules, we first branch on neighbors of leaves in $G[\widehat{U}]$ whose degree in this subgraph is at least two, then on leaves in $G[\widehat{U}]$, then of vertices of degree at least three in $G[\widehat{U}]$, and finally (in the last two rules) on the remaining vertices in $G[\widehat{U}]$ that are not isolated in this subgraph. Although we can merge some of them, we present them separately for the sake of clarity.

\begin{branchrule}\label{rule:wvcneileaf}
{\normalfont [There are $v,u\in \widehat{U}$ such that $N(u)\cap \widehat{U} = \{v\}$ and $|N(v)\cap \widehat{U}|\geq 2$] Return the set of maximum weight among $A$ and $B$, computed in the following branches.
\begin{enumerate}
\item $A\Leftarrow$\alg{ProcedureB}($G,w,U,I\cup\{v\},O$).
\item $B\Leftarrow$\alg{ProcedureB}($G,w,U,I\cup (N(v)\cap \widehat{U}),O\cup\{v\}$).
\end{enumerate}}
\end{branchrule}

{\noindent The correctness of Lemma \ref{lemma:presProB} is preserved since every vertex cover contains either $v$ (an option examined in the first branch) or does not contain $v$, in which case it must contain all the neighbors of $v$ (an option examined in the second branch).}

We get a branching vector that is at least as good as $(|\{v,u\}|,|(N(v)\cap \widehat{U})\cup\{v\}|)$. Indeed, in the first branch, $v$ is inserted to $U$ and $u$ is inserted to $S(I\cup O)$, and in the second branch, $(N(v)\cap \widehat{U})$ is inserted to $I$ and $v$ is inserted to $O$. Since this branching vector is at least as good as $(2,3)$ (because $|N(v)\cap \widehat{U}|\geq 2$), we get a root that is at most $3^{\frac{1}{3}}$.

\begin{branchrule}\label{rule:wvcneileaf2}
{\normalfont [There are $v,u\in \widehat{U}$ such that $N(u)\cap \widehat{U} = \{v\}$] Return the set of maximum weight among $A$ and $B$, computed in the following branches.
\begin{enumerate}
\item $A\Leftarrow$\alg{ProcedureB}($G,w,U,I\cup\{v\},O$).
\item $B\Leftarrow$\alg{ProcedureB}($G,w,U,I\cup (N(v)\cap \widehat{U}),O\cup\{v\}$).
\end{enumerate}}
\end{branchrule}

{\noindent For correctness, follow the explanation given in the previous rule. Also, since the previous rule did not apply, $N(v)\cap \widehat{U}=\{u\}$. Thus, at each branch, one vertex in $\{v,u\}$ is inserted to $I$, and the other is inserted to $S(I\cup O)\cup O$. We get the branching vector $(|\{v,u\}|,|\{v,u\}|)=(2,2)$, whose root is at most $3^{\frac{1}{3}}$. We note that we did not merge this rule with the previous one, since in the last rule we use the fact that Rule \ref{rule:wvcneileaf} has a branching vector better than $(2,2)$.}

\begin{branchrule}\label{rule:tvc3}
{\normalfont [There is $v\in \widehat{U}$ such that $|N(v)\cap \widehat{U}| \geq 3$] Return the set of maximum weight among $A$ and $B$, computed in the following branches.
\begin{enumerate}
\item $A\Leftarrow$\alg{ProcedureB}($G,w,U,I\cup\{v\},O$).
\item $B\Leftarrow$\alg{ProcedureB}($G,w,U,I\cup (N(v)\cap \widehat{U}),O\cup\{v\}$).
\end{enumerate}}
\end{branchrule}

{\noindent For correctness, follow the explanation given in Rule \ref{rule:wvcneileaf}. Clearly, we get the branching vector $(|\{v\}|,|(N(v)\cap \widehat{U})\cup\{v\}|)$. This branching vector is at least as good as $(1,4)$ (since $|N(v)\cap \widehat{U}| \geq 3$), and thus we get a root that is at most $3^{\frac{1}{3}}$.}

\begin{branchrule}
{\normalfont [There are $v,u,r\in \widehat{U}$ such that $N(v)\cap \widehat{U} = \{u,r\}$, $N(u)\cap \widehat{U} = \{v,r\}$ and $N(r)\cap \widehat{U} = \{v,u\}$] Return the set of maximum weight among $A$, $B$ and $C$, computed in the following branches.
\begin{enumerate}
\item $A\Leftarrow$\alg{ProcedureB}($G,w,U,I\cup\{v,u\},O$).
\item $B\Leftarrow$\alg{ProcedureB}($G,w,U,I\cup\{v,r\},O$).
\item $C\Leftarrow$\alg{ProcedureB}($G,w,U,I\cup\{u,r\},O$).
\end{enumerate}
}\end{branchrule}

{\noindent The correctness of Lemma \ref{lemma:presProB} is preserved since every vertex cover contains at least two vertices from $\{v,u,r\}$, thus it contains $v,u$ (an option examined in the first branch),\footnote{It might also contain $r$: once we insert $v,u$ to $I$, we do {\em not} insert $r$ to $O$.} or it contains $v,r$ (an option examined in the second branch), or it contains $u,r$ (an option examined in the third branch). At each branch, two vertices of the triangle are inserted to $I$, and the other one is inserted to $S(I\cup O)$. Thus, we get the branching vector $(3,3,3)$, whose root is $3^{\frac{1}{3}}$.}

\begin{branchrule}
{\normalfont Let $v$ be a vertex in $\widehat{U}$ such that $|N(v)\cap \widehat{U}|=2$. Return the set of maximum weight among $A$ and $B$, computed in the following branches.
\begin{enumerate}
\item $A\Leftarrow$\alg{ProcedureB}($G,w,U,I\cup\{v\},O$).
\item $B\Leftarrow$\alg{ProcedureB}($G,w,U,I\cup (N(v)\cap \widehat{U}),O\cup\{v\}$).
\end{enumerate}}
\end{branchrule}

{\noindent For correctness, follow the explanation given in Rule \ref{rule:wvcneileaf}. In this rule, $G[\widehat{U}]$ does not contain vertices of degree at least three (due to Rule \ref{rule:tvc3}), triangles (due to the previous rule) or leaves (due to Rules \ref{rule:wvcneileaf} and \ref{rule:wvcneileaf2}). Thus, the connected components in $G[\widehat{U}]$ are only cycles, each on at least four vertices. Thus, after inserting $v$ to $I$ (in the first branch), we can next apply Rule \ref{rule:wvcneileaf}. This results in a branching vector that is at least as good as $(1+(2,3),3)=(3,4,3)$, whose root is at most~$3^{\frac{1}{3}}$.}

\end{document}